\documentclass[11pt]{article}

\usepackage[a4paper,margin=1in]{geometry}
\usepackage{amsmath}
\usepackage{amssymb}
\usepackage{amsthm}
\usepackage{hyperref}
\usepackage{authblk}

\newtheorem{lemma}{Lemma}[section]

\newtheorem{corollary}[lemma]{Corollary}
\newtheorem{proposition}[lemma]{Proposition}

\title{
Stabilization Without Simplification:\\
A Two-Dimensional Model of Software Evolution
}

\author[1]{Masaru Furukawa}
\affil[1]{Professor Emeritus, University of Toyama, Japan}

\date{}

\begin{document}

\maketitle

% =========================
\begin{abstract}
Software systems are widely observed to grow in size, complexity, and interdependence over time, yet many large-scale systems remain stable despite persistent structural burden. 
This apparent tension suggests a limitation in one-dimensional views of software evolution.

This paper introduces a graph-based, discrete-time probabilistic framework that separates structural burden from uncertainty. 
Change effort is modeled as a stochastic variable determined by the dependency neighborhood of the changed entity and by residual variability. 
Within this framework, burden is defined as expected effort and uncertainty as variance of effort.

We show that, on a time interval in which structural exposure does not decrease while structural dispersion, residual dispersion, or their covariance is reduced, a software system can exhibit a local stabilization-without-simplification regime: expected change burden remains non-decreasing, while uncertainty decreases. 
The result characterizes stabilization without simplification as a distributional regime rather than as a global monotonic property of software evolution.

The proposed framework provides a minimal theoretical explanation for how software systems can become more predictable over time without necessarily becoming structurally simpler, and offers a foundation for further theoretical and empirical studies of software evolution.
\end{abstract}

% =========================
\section{Introduction}

Software systems are widely observed to grow in size, complexity, and interdependence over time. 
This phenomenon has been extensively documented in the literature on software evolution~\cite{lehman1997,mens2008}, where increasing complexity is often associated with higher maintenance cost and reduced comprehensibility.

At the same time, large-scale software systems rarely collapse under their own complexity. 
Instead, many mature systems continue to evolve in a stable and manageable manner despite substantial structural burden. 
This apparent tension raises a fundamental question: how can software systems become more stable while remaining structurally complex?

A common implicit assumption in software engineering is that stabilization is achieved through simplification, such as reducing dependencies, improving modularity, or eliminating unnecessary complexity. 
Although such mechanisms are clearly important, they do not fully explain why many systems remain viable even when their structural load persists or increases over time.

This suggests a limitation in existing one-dimensional views of software evolution. 
Most approaches measure evolution primarily in terms of size, complexity, or defect-related indicators, and therefore tend to treat increasing complexity as inherently detrimental. 
However, such views do not distinguish between two different aspects of change: the expected cost of modifying a system and the variability of that cost.

In this paper, we propose a two-dimensional theoretical framework that separates these aspects into \emph{structural burden} and \emph{uncertainty}. 
Structural burden represents the expected cost of change, while uncertainty represents the variability or unpredictability of change effort. 
By distinguishing these dimensions, software evolution can be modeled as a trajectory in a two-dimensional state space.

The central claim of this paper is as follows:

\begin{quote}
Software evolution can exhibit a systematic decoupling between structural burden and uncertainty, in which uncertainty decreases while structural burden does not necessarily decrease.
\end{quote}

To formalize this claim, we model software evolution in discrete time and represent system structure by a dependency graph. 
Change effort is modeled as a random variable whose value depends on the structural neighborhood of the changed entity and on residual stochastic effects. 
Within this framework, burden is defined as expected effort and uncertainty as variance of effort.

The main contribution of this paper is theoretical. 
We show that, on a time interval over which explicit structural and stochastic assumptions hold, there is a local regime in which uncertainty decreases while burden remains constant or increases. 
This result establishes that \emph{stabilization without simplification} is not merely an empirical observation, but a structurally admissible consequence of the model.

The contributions of this paper are threefold:

\begin{itemize}
\item \textbf{Conceptual contribution:} We distinguish structural burden from uncertainty as two separate dimensions of software evolution.
\item \textbf{Formal contribution:} We introduce a discrete-time graph-based probabilistic framework in which burden and uncertainty are defined and analyzed rigorously.
\item \textbf{Theoretical contribution:} We prove that stabilization without simplification arises under explicit assumptions on structural load, structural regularization, process stabilization, and covariance control.

\end{itemize}

This paper does not attempt to establish broad empirical generality or to identify all causal mechanisms of software evolution. 
Rather, it provides a minimal theoretical foundation for understanding a recurring phenomenon observed in evolving software systems.

The remainder of this paper is organized as follows. 
Section~2 introduces the conceptual background. 
Section~3 presents the formal setup. 
Section~4 defines the core quantities of the model. 
Section~5 derives structural properties. 
Section~6 interprets the resulting dynamics. 
Section~7 relates the framework to empirical observations. 
Section~8 concludes with implications and limitations.
\section{Conceptual Background}

\subsection{Software Evolution and Structural Growth}

Software systems are widely observed to grow in size, complexity, and interdependence over time~\cite{lehman1997,mens2008}. 
As systems evolve, they accumulate dependencies, interfaces, coordination requirements, and other structural constraints that increase the cost of modification.

This structural growth is a central theme in software evolution research. 
However, it does not by itself explain why many mature systems remain viable and manageable despite continued complexity growth.

\subsection{Stability Beyond Simplification}

A common view in software engineering is that stabilization is achieved through simplification, for example by reducing coupling, improving modularity, or removing unnecessary complexity. 
These mechanisms are clearly important, but they do not exhaust the possible ways in which a system may become stable.

In practice, systems may remain structurally demanding while becoming easier to change in a predictable manner. 
This suggests that stability should not be understood solely as a reduction in structural load, but also as a reduction in the unpredictability of change outcomes.

\subsection{Magnitude and Variability of Change}

To understand this distinction, it is useful to separate two aspects of change effort.

First, changes have an expected magnitude: some systems require more effort on average because their structural dependencies make modification costly. 
Second, changes have variability: even if the average cost is high, the actual cost of individual changes may become more or less predictable over time.

Most traditional one-dimensional views of software evolution do not distinguish these two aspects explicitly. 
Instead, they compress them into a single notion of complexity or maintenance difficulty.

\subsection{Toward a Two-Dimensional View}

The present work adopts a two-dimensional perspective in which software evolution is described by:

\begin{itemize}
\item \textbf{Structural burden:} the expected cost of change
\item \textbf{Uncertainty:} the variability or unpredictability of change effort
\end{itemize}

These quantities capture different properties of evolving software systems. 
A system may exhibit high structural burden but low uncertainty if its dependencies are extensive yet well understood. 
Conversely, a system may exhibit low structural burden but high uncertainty if changes are small on average but difficult to predict.

This distinction makes it possible to describe software evolution as a trajectory in a two-dimensional state space rather than as movement along a single axis of increasing or decreasing complexity.

\subsection{Why a Formal Model is Needed}

The conceptual distinction between burden and uncertainty is intuitive, but intuition alone is not sufficient. 
To make this distinction analytically useful, a formal framework is needed.

In particular, a formal model should answer the following question:

\begin{quote}
Under what structural and stochastic conditions can uncertainty decrease while structural burden does not?
\end{quote}

This question cannot be answered adequately within a purely one-dimensional framework. 
It requires an explicit representation of both the structural sources of effort and the stochastic sources of variability.

\subsection{Conceptual Gap Addressed in This Paper}

The main conceptual gap addressed here is the absence of a formal separation between the expected cost of change and the variability of that cost.

Without such a separation, it is difficult to explain how systems can remain stable despite persistent or increasing structural complexity. 
The framework developed in this paper addresses this gap by introducing a graph-based, discrete-time probabilistic model in which these two dimensions are explicitly defined and analyzed.

This provides the foundation for the main theoretical result of the paper: stabilization without simplification is a structurally admissible regime of software evolution.
\section{Formal Setup}

\subsection{Discrete-Time Evolution}

We consider software evolution over discrete time steps
\[
t = 0,1,2,\dots .
\]
At each time step, the software system is represented by a structural state $S_t$.

To make the notion of structure explicit, we model $S_t$ as a directed dependency graph:

\begin{equation}
G_t = (V_t, E_t)
\end{equation}

where $V_t$ is the set of software entities (e.g., files, modules, or components) and $E_t \subseteq V_t \times V_t$ is the set of dependency relations.

\subsection{Change Events}

At time $t$, a change event selects a target entity
\[
X_t \in V_t
\]
and induces a modification effort $e_t$.

The target entity is treated as a random variable drawn from a probability distribution
\[
p_t(v) = \Pr(X_t = v), \qquad v \in V_t .
\]

Thus, software evolution is modeled as a sequence of random change events occurring on a time-varying dependency graph.

\subsection{Local Structural Load}

For a node $v \in V_t$, let

\begin{equation}
N_t^{+}(v) = \{ u \in V_t \mid (v,u) \in E_t \}
\end{equation}

denote its outgoing dependency neighborhood, and define its out-degree by

\begin{equation}
d_t(v) = |N_t^{+}(v)|.
\end{equation}

We interpret $d_t(v)$ as a local measure of structural exposure: the larger the outgoing dependency neighborhood, the greater the local dependency load associated with changing $v$. This is a deliberately minimal first-order measure. Depending on how dependency direction is encoded, propagation may alternatively be represented by incoming dependencies, reverse reachability, or transitive dependency measures; these are treated as extensions rather than as part of the minimal model.

\subsection{Effort Model}

We assume that the effort of a change event at time $t$ is given by

\begin{equation}
e_t = \alpha\, d_t(X_t) + \beta + \epsilon_t
\end{equation}

where $\alpha > 0$ and $\beta \geq 0$ are constants, and $\epsilon_t$ is a stochastic term satisfying only the unconditional mean-zero condition

\begin{equation}
\mathbb{E}[\epsilon_t] = 0.
\end{equation}

We do not assume conditional mean independence of $\epsilon_t$ from $X_t$. Thus, the residual component may covary with the structural exposure $d_t(X_t)$, and this covariance is captured explicitly below by $c_t = \mathrm{Cov}(d_t(X_t), \epsilon_t)$. The linear effort model is an idealized moment model; non-negativity of realized effort can be imposed by support restrictions on $\epsilon_t$ or by interpreting $e_t$ as a local linear approximation to a non-negative effort scale.

Under this model, effort consists of two components:

\begin{itemize}
\item a structural component determined by the dependency load of the changed entity, and
\item a stochastic component representing residual variability not directly explained by the graph structure.
\end{itemize}

\subsection{Burden and Uncertainty}

We define two aggregate quantities at time $t$:

\begin{equation}
B_t = \mathbb{E}[e_t]
\end{equation}

and

\begin{equation}
U_t = \mathrm{Var}(e_t).
\end{equation}

Here, the expectation and variance are taken over both the random target selection $X_t$ and the stochastic fluctuation $\epsilon_t$.

The quantity $B_t$ represents structural burden, while $U_t$ represents uncertainty.

\subsection{Difference Operators}

Because time is discrete, dynamic change is represented by first differences rather than derivatives:

\begin{equation}
\Delta B_t = B_{t+1} - B_t,
\qquad
\Delta U_t = U_{t+1} - U_t.
\end{equation}

In this framework:

\begin{itemize}
\item stabilization corresponds to $\Delta U_t < 0$,
\item non-simplification corresponds to $\Delta B_t \geq 0$.
\end{itemize}

\subsection{Derived Structural Quantities}

Let

\begin{equation}
\mu_t = \mathbb{E}[d_t(X_t)]
\end{equation}

denote the expected structural load of changed entities, and let

\begin{equation}
\sigma_{d,t}^2 = \mathrm{Var}(d_t(X_t))
\end{equation}

denote the variance of structural load across changed entities.

We also define the residual variance

\begin{equation}
\sigma_{\epsilon,t}^2 = \mathrm{Var}(\epsilon_t)
\end{equation}

and the covariance term

\begin{equation}
c_t = \mathrm{Cov}(d_t(X_t), \epsilon_t).
\end{equation}

These quantities allow burden and uncertainty to be expressed in terms of structural and stochastic components.

\subsection{Assumptions}

To derive nontrivial results, we impose the following assumptions.

\begin{itemize}
\item \textbf{A1 (Non-decreasing average structural load).} The expected structural load of changed entities does not decrease:
\begin{equation}
\mu_{t+1} \geq \mu_t.
\end{equation}

\item \textbf{A2 (Structural regularization).} The variance of structural load does not increase:
\begin{equation}
\sigma_{d,t+1}^2 \leq \sigma_{d,t}^2.
\end{equation}

\item \textbf{A3 (Process stabilization).} The residual variance does not increase:
\begin{equation}
\sigma_{\epsilon,t+1}^2 \leq \sigma_{\epsilon,t}^2.
\end{equation}

\item \textbf{A4 (Covariance control).} The covariance between structural load and residual effort does not increase:
\begin{equation}
c_{t+1} \leq c_t.
\end{equation}
\end{itemize}

These conditions are imposed over a specified time interval rather than over the entire lifetime of a system. Since $U_t = \mathrm{Var}(e_t) \geq 0$ and the covariance term is bounded by the Cauchy--Schwarz inequality,
\begin{equation}
|c_t| \leq \sigma_{d,t}\sigma_{\epsilon,t},
\end{equation}
the moment conditions on $\sigma_{d,t}^2$, $\sigma_{\epsilon,t}^2$, and $c_t$ are not independent. Assumption A4 is therefore imposed only within the admissible region defined by these variance and covariance constraints. Equivalently, one may write
\begin{equation}
\rho_t = \frac{c_t}{\sigma_{d,t}\sigma_{\epsilon,t}} \in [-1,1],
\end{equation}
and interpret covariance control as a non-increasing association between structural exposure and residual effort within this feasible region.
These assumptions do not imply that the system becomes simpler. 
Rather, they state that while average structural exposure may persist or grow, structural heterogeneity, residual fluctuation, and the dependence between structural load and residual effort may contract over time.

\subsection{Scope of the Setup}

The present setup is intentionally limited. 
It does not attempt to model all aspects of software systems, such as semantic structure, developer interaction, or organizational dynamics.

Instead, it provides a minimal graph-based probabilistic foundation for reasoning about how structural burden and uncertainty may evolve differently over time. 
The next section builds on this setup to define the core concepts of the Penalty of Change (POC) framework, where software evolution is represented through the joint movement of burden and uncertainty.
\section{Core Definitions}

Based on the formal setup introduced in the previous section, we now define the core quantities of the Penalty of Change (POC) framework.

\subsection{Structural Burden}

We define the structural burden at time $t$ as the expected effort of change:

\begin{equation}
B_t = \mathbb{E}[e_t].
\end{equation}

Under the effort model
\[
e_t = \alpha\, d_t(X_t) + \beta + \epsilon_t,
\]
and the unconditional mean-zero assumption
\[
\mathbb{E}[\epsilon_t] = 0,
\]
it follows that

\begin{equation}
B_t = \alpha \mu_t + \beta,
\end{equation}

where
\[
\mu_t = \mathbb{E}[d_t(X_t)].
\]

Thus, burden is determined by the average structural load of the entities selected for change.

\subsection{Uncertainty}

We define uncertainty at time $t$ as the variance of change effort:

\begin{equation}
U_t = \mathrm{Var}(e_t).
\end{equation}

Under the effort model, uncertainty can be written as

\begin{equation}
U_t = \alpha^2 \sigma_{d,t}^2 + \sigma_{\epsilon,t}^2 + 2\alpha c_t,
\end{equation}

where
\[
\sigma_{d,t}^2 = \mathrm{Var}(d_t(X_t)),
\qquad
\sigma_{\epsilon,t}^2 = \mathrm{Var}(\epsilon_t),
\qquad
c_t = \mathrm{Cov}(d_t(X_t), \epsilon_t).
\]

Accordingly, uncertainty reflects structural heterogeneity, residual process variability, and their covariance. The covariance term is not forced to be zero because the model does not assume conditional mean independence of residual effort from the selected change target.

As a genuine uncorrelated special case, if $c_t = 0$, then

\begin{equation}
U_t = \alpha^2 \sigma_{d,t}^2 + \sigma_{\epsilon,t}^2.
\end{equation}

\subsection{Decoupling}

We define decoupling between burden and uncertainty as the condition in which their temporal evolutions are not locked to the same direction.

Formally, decoupling is said to occur at time $t$ if

\begin{equation}
\Delta U_t < 0
\qquad \text{while} \qquad
\Delta B_t \geq 0.
\end{equation}

This captures the regime in which software systems become more predictable without becoming structurally simpler.

\subsection{Stabilization}

We define stabilization as a decrease in uncertainty over time:

\begin{equation}
\Delta U_t < 0.
\end{equation}

Under this definition, stabilization refers to increasing predictability in change effort rather than to reduced size or lower structural load.

\subsection{Non-Simplification}

We define non-simplification as the absence of a decrease in burden:

\begin{equation}
\Delta B_t \geq 0.
\end{equation}

This allows structural burden to remain constant or increase over time.

\subsection{Stabilization Without Simplification}

We define \emph{stabilization without simplification} as the regime in which both stabilization and non-simplification hold:

\begin{equation}
\Delta U_t < 0
\qquad \text{and} \qquad
\Delta B_t \geq 0.
\end{equation}

This is the central local regime analyzed in the remainder of the paper. The condition is understood over intervals where the relevant structural and stochastic assumptions hold; it is not asserted as a global monotonic law over the entire lifetime of a software system.

\subsection{POC State}

For each time step $t$, we define the state of the system in the POC framework by the ordered pair

\begin{equation}
\mathcal{P}_t = (B_t, U_t).
\end{equation}

Software evolution is then represented as a trajectory of states
\[
\mathcal{P}_0, \mathcal{P}_1, \mathcal{P}_2, \dots
\]
in the burden--uncertainty plane.

\subsection{Remarks}

These definitions are intentionally minimal. 
They rely only on the first- and second-order properties of change effort and do not depend on a specific programming language, architectural style, or development process.

The purpose of the framework is not to exhaustively model software evolution, but to isolate a structural distinction between burden and uncertainty that can support rigorous analysis.
\section{Structural Properties}

We now derive structural consequences of the graph-based effort model introduced in Section~3 and the core definitions in Section~4.

\subsection{Burden and Uncertainty as Derived Quantities}

We begin by making explicit how burden and uncertainty are determined by the structural and stochastic components of the model.

\begin{lemma}[Burden Formula]\label{lem:burden}
Under the effort model
\begin{equation}
e_t = \alpha\, d_t(X_t) + \beta + \epsilon_t,
\end{equation}
with $\mathbb{E}[\epsilon_t] = 0$, the structural burden satisfies
\begin{equation}
B_t = \alpha \mu_t + \beta,
\end{equation}
where $\mu_t = \mathbb{E}[d_t(X_t)]$.
\end{lemma}

\textit{Proof.}
By definition,
\[
B_t = \mathbb{E}[e_t].
\]
Substituting the effort model gives
\[
B_t = \mathbb{E}[\alpha d_t(X_t) + \beta + \epsilon_t]
     = \alpha \mathbb{E}[d_t(X_t)] + \beta + \mathbb{E}[\epsilon_t].
\]
Using $\mathbb{E}[\epsilon_t]=0$, we obtain
\[
B_t = \alpha \mu_t + \beta.
\]
\hfill $\square$

\begin{lemma}[General Uncertainty Formula]\label{lem:uncertainty}
Under the effort model,
\begin{equation}
U_t
=
\alpha^2 \sigma_{d,t}^2
+
\sigma_{\epsilon,t}^2
+
2\alpha\, c_t,
\end{equation}
where
\[
\sigma_{d,t}^2 = \mathrm{Var}(d_t(X_t)),
\qquad
\sigma_{\epsilon,t}^2 = \mathrm{Var}(\epsilon_t),
\qquad
c_t = \mathrm{Cov}(d_t(X_t), \epsilon_t).
\]
\end{lemma}

\textit{Proof.}
By definition,
\[
U_t = \mathrm{Var}(e_t).
\]
Substituting the effort model,
\[
U_t = \mathrm{Var}(\alpha d_t(X_t) + \beta + \epsilon_t).
\]
Since $\beta$ is constant,
\[
U_t
=
\alpha^2 \mathrm{Var}(d_t(X_t))
+
\mathrm{Var}(\epsilon_t)
+
2\alpha\,\mathrm{Cov}(d_t(X_t),\epsilon_t).
\]
Hence,
\[
U_t
=
\alpha^2 \sigma_{d,t}^2
+
\sigma_{\epsilon,t}^2
+
2\alpha\, c_t.
\]
\hfill $\square$

\begin{corollary}[Uncorrelated Special Case]
If $c_t = \mathrm{Cov}(d_t(X_t),\epsilon_t)=0$, then
\begin{equation}
U_t = \alpha^2 \sigma_{d,t}^2 + \sigma_{\epsilon,t}^2.
\end{equation}
\end{corollary}

\textit{Proof.}
Immediate from Lemma~\ref{lem:uncertainty}.
\hfill $\square$

These results show that burden depends only on the average structural load of changed entities, whereas uncertainty depends on three components: structural heterogeneity, residual stochastic fluctuation, and their covariance.

\subsection{Monotonicity Under Structural Regularization and Covariance Control}

We now characterize the local dynamic regime implied by the assumptions introduced in Section~3.

\begin{proposition}[Characterization of a Local SWS Regime]\label{prop:local-sws}
Consider a time interval $[t_0,t_1]$ over which assumptions A1--A4 hold. That is, for each $t \in [t_0,t_1)$,

\begin{itemize}
\item[\textbf{A1}] $\mu_{t+1} \geq \mu_t$
\item[\textbf{A2}] $\sigma_{d,t+1}^2 \leq \sigma_{d,t}^2$
\item[\textbf{A3}] $\sigma_{\epsilon,t+1}^2 \leq \sigma_{\epsilon,t}^2$
\item[\textbf{A4}] $c_{t+1} \leq c_t$, where $c_t = \mathrm{Cov}(d_t(X_t), \epsilon_t)$, within the admissible covariance region.
\end{itemize}

Then, for each $t \in [t_0,t_1)$,
\begin{equation}
\Delta B_t \geq 0
\end{equation}
and
\begin{equation}
\Delta U_t \leq 0.
\end{equation}

If A1 is strict for a transition, then burden strictly increases over that transition. If at least one of A2--A4 is strict for a transition, then uncertainty strictly decreases over that transition. Hence, when burden is non-decreasing while uncertainty is decreasing on the interval, the system exhibits a local stabilization-without-simplification regime.
\end{proposition}

\textit{Proof.}
From Lemma~\ref{lem:burden},
\[
B_t = \alpha \mu_t + \beta.
\]
Therefore,
\[
\Delta B_t = B_{t+1} - B_t = \alpha(\mu_{t+1} - \mu_t).
\]
Since $\alpha > 0$ and A1 states that $\mu_{t+1} \geq \mu_t$, it follows that
\[
\Delta B_t \geq 0.
\]

From Lemma~\ref{lem:uncertainty},
\[
U_t = \alpha^2 \sigma_{d,t}^2 + \sigma_{\epsilon,t}^2 + 2\alpha c_t.
\]
Hence,
\[
\Delta U_t
=
U_{t+1} - U_t
=
\alpha^2(\sigma_{d,t+1}^2 - \sigma_{d,t}^2)
+
(\sigma_{\epsilon,t+1}^2 - \sigma_{\epsilon,t}^2)
+
2\alpha(c_{t+1} - c_t).
\]
By A2, A3, and A4, each term on the right-hand side is non-positive, so
\[
\Delta U_t \leq 0.
\]
If at least one of these inequalities is strict, then $\Delta U_t < 0$.
\hfill $\square$

\subsection{Corollaries}

The proposition immediately yields several useful consequences.

\begin{corollary}[Decoupling is Structurally Admissible]
On any interval over which assumptions A1--A4 hold, burden and uncertainty need not evolve in the same direction. In particular, if at least one of A2--A4 is strict, burden may remain constant or increase while uncertainty decreases.
\end{corollary}

\textit{Proof.}
Immediate from Proposition~\ref{prop:local-sws}.
\hfill $\square$

\begin{corollary}[Simplification is Not Necessary for Stabilization]
On any interval over which assumptions A1--A4 hold, stabilization can occur without any reduction in expected structural burden, provided at least one uncertainty-reducing condition is strict.
\end{corollary}

\textit{Proof.}
By Proposition~\ref{prop:local-sws}, $\Delta U_t < 0$ may hold while $\Delta B_t \geq 0$.
Thus, stabilization does not require $\Delta B_t < 0$.
\hfill $\square$

\subsection{Interpretation}

Proposition~\ref{prop:local-sws} should be read as a characterization result rather than as a deep existence theorem. The direction of $B_t$ follows from its affine dependence on $\mu_t$, while the direction of $U_t$ follows from the variance decomposition into structural dispersion, residual dispersion, and covariance. The contribution of the result is therefore not that stabilization without simplification emerges from arbitrary conditions, but that the regime can be represented as a specific configuration of distributional moments:

\begin{itemize}
\item non-decreasing average structural exposure of changed entities,
\item non-increasing heterogeneity of structural exposure,
\item non-increasing residual process variability, and
\item non-increasing dependence between structural exposure and residual effort.
\end{itemize}

Thus, the burden--uncertainty decoupling is not merely a verbal definition. It is represented by interpretable moment conditions: non-decreasing expected structural exposure together with non-increasing structural dispersion, residual process variability, and admissible covariance. When at least one uncertainty component strictly decreases, this representation yields a local stabilization-without-simplification regime.
\section{Dynamic Interpretation}

The proposition established in the previous section characterizes stabilization without simplification as a local regime that follows from explicit structural and stochastic conditions. 
We now interpret the meaning of these conditions in software-evolution terms.

\subsection{From Static Definitions to Dynamic Regimes}

The quantities $B_t$ and $U_t$ define the state of the system at time $t$ in the burden--uncertainty plane. 
By itself, this state representation is static. 
Proposition~\ref{prop:local-sws} introduces dynamics by showing how the evolution of structural and stochastic components determines the direction of motion in this space over intervals where A1--A4 hold.

In particular, the regime
\[
\Delta B_t \geq 0,
\qquad
\Delta U_t < 0
\]
describes systems whose expected structural cost of change does not decrease, while the variability of change effort contracts over time.

\subsection{Interpretation of Assumption A1}

Assumption A1 states that the expected structural load of changed entities does not decrease:
\[
\mu_{t+1} \geq \mu_t.
\]

This means that the average dependency exposure of the locations being modified remains constant or increases. 
In software terms, the system does not become structurally simpler. 
Changes continue to involve entities embedded in nontrivial dependency neighborhoods.

This assumption reflects the common situation in which mature software systems continue to accumulate interfaces, coordination requirements, and architectural constraints.

\subsection{Interpretation of Assumption A2}

Assumption A2 states that the variance of structural load across changed entities does not increase:
\[
\sigma_{d,t+1}^2 \leq \sigma_{d,t}^2.
\]

This condition admits at least two interpretations.

First, the structure itself may become more regular, so that dependency exposure is less heterogeneous across entities. 
In this interpretation, the graph becomes easier to navigate because structurally irregular regions become less dominant.

Second, even if the graph remains heterogeneous, developers and processes may learn to avoid, isolate, or better manage the most irregular regions. 
In that case, the effective selection distribution over changed targets becomes more regular, even without a major simplification of the underlying graph.

In probabilistic terms, this second interpretation means that the selection distribution over changed entities becomes more concentrated on structurally better-behaved regions, even if the underlying graph itself remains heterogeneous.

Thus, A2 can reflect either structural regularization of the system itself or regularization in how the system is engaged through change.

\subsection{Interpretation of Assumption A3}

Assumption A3 states that the residual variance does not increase:
\[
\sigma_{\epsilon,t+1}^2 \leq \sigma_{\epsilon,t}^2.
\]

This residual term captures variability not directly explained by the local dependency structure. 
Its reduction can be interpreted as \emph{process stabilization}, arising from factors such as accumulated developer knowledge, repeated change patterns, testing infrastructure, code review discipline, and workflow standardization.

Thus, even when structural burden persists, change effort can become more predictable because the non-structural sources of variability are reduced.

\subsection{Interpretation of Assumption A4}

Assumption A4 states that the covariance term does not increase:
\[
c_{t+1} \leq c_t,
\qquad
c_t = \mathrm{Cov}(d_t(X_t), \epsilon_t).
\]

This condition controls the degree to which structurally heavy changes are systematically associated with unusually large residual effort.

A positive covariance means that changes in structurally exposed regions also carry additional unexplained difficulty. 
A reduction in this covariance indicates that structurally complex changes become less exceptional in their residual behavior. 
In practical terms, this can arise when teams learn how to handle complex regions more routinely, reducing the extra unpredictability previously associated with them.

Thus, A4 prevents the covariance term from offsetting the stabilizing effects of A2 and A3.

\subsection{The Dynamic Mechanism of Stabilization Without Simplification}

Taken together, assumptions A1--A4 imply a specific dynamic mechanism.

\begin{itemize}
\item The system continues to operate under nontrivial or increasing structural load.
\item Structural exposure becomes less heterogeneous or is engaged more regularly through change.
\item Residual process variability contracts.
\item The dependence between structural heaviness and residual difficulty weakens or at least does not intensify.
\end{itemize}

Under these conditions, the expected effort of change does not decrease, and the variance of effort does not increase; when at least one uncertainty-reducing condition is strict, the variance of effort decreases. 
The resulting motion in the burden--uncertainty plane is therefore directionally characterized rather than random:
\[
(B_t, U_t) \longrightarrow (\text{non-decreasing burden}, \text{decreasing uncertainty}).
\]

This is the formal meaning of stabilization without simplification.

\subsection{Graph-Theoretic Interpretation}

In graph-theoretic terms, the proposition implies that software evolution need not reduce the average dependency exposure of changed nodes in order to stabilize locally. 
Instead, stabilization can arise because

\begin{itemize}
\item the dispersion of dependency exposure contracts,
\item the residual variability of change effort decreases, and
\item structurally exposed changes cease to generate disproportionately unpredictable residual effort.
\end{itemize}

Thus, software evolution does not require simplification of the dependency graph. 
What matters is not whether the graph becomes smaller or less connected, but whether the distribution of change effort becomes more concentrated and better behaved.

\subsection{Internalization of Complexity}

A useful interpretation of the above mechanism is that complexity is progressively internalized rather than eliminated.

As systems evolve, developers, tools, and processes adapt to structural complexity. 
What was initially uncertain becomes routinized; what was initially heterogeneous becomes more regular; and what was initially difficult in an exceptional way becomes more normalizable.

The system remains complex, but that complexity is increasingly absorbed into stable operational patterns. 
This perspective explains how mature systems can remain viable even when their structural burden does not decline.

\subsection{On the Locality of the Structural Measure}

The present model uses first-order local dependency exposure, represented by the out-degree of the changed entity, as a baseline structural measure. 
This choice is intentional: it provides the simplest graph-based notion of local structural exposure that supports rigorous analysis. It should not be read as the only possible direction-sensitive measure of propagation cost.

More global notions of structural burden, such as transitive reachability, centrality, or multilayer dependency measures, are natural extensions of the framework. 
Such extensions may capture broader forms of propagation overhead and could yield stronger or more refined dynamic results.

Accordingly, the present model should be understood as a minimal graph-based theory rather than as a complete account of software structure.

\subsection{Implications for the Theory of Software Evolution}

The dynamic interpretation developed here suggests that software evolution should be modeled as a two-dimensional process rather than as a one-dimensional trend toward either increasing or decreasing complexity.

One-dimensional views cannot adequately capture the coexistence of persistent burden and decreasing uncertainty. 
By contrast, the burden--uncertainty framework shows that structural complexity and predictability may follow different trajectories and that their decoupling is theoretically meaningful.

This provides a more precise foundation for analyzing how software systems evolve toward stability.
\section{Relation to Empirical Observations}

The theoretical framework developed in this paper is motivated by recurring empirical patterns observed in studies of software evolution. 
In particular, analyses of large-scale open-source systems suggest that mature software projects often become more predictable over time even while remaining structurally complex~\cite{mockus2002,bird2009,kalliamvakou2014}.

The purpose of this section is not to provide a full empirical validation of the model, but to clarify how the quantities and assumptions introduced in the previous sections relate to observable behavior in real software systems.

\subsection{Empirical Interpretation of Burden}

In the present framework, burden is defined as the expected change effort:
\[
B_t = \mathbb{E}[e_t].
\]

Empirically, this quantity is connected to the average cost of change through observable proxies such as the number of changed files, lines modified, or other indicators of change size and coordination demand. This connection is a measurement bridge rather than a direct operationalization: the model primitive $d_t(X_t)$ is a degree-based structural exposure measure, whereas empirical proxies such as files changed are change-size or coordination-demand indicators.

The condition
\[
\Delta B_t \geq 0
\]
therefore corresponds to the empirical situation in which the average cost of change does not decrease over time. 
This is consistent with observations from large systems whose change processes remain nontrivial despite continued maturation.

\subsection{Empirical Interpretation of Uncertainty}

Uncertainty is defined as the variance of change effort:
\[
U_t = \mathrm{Var}(e_t).
\]

Empirically, this quantity is connected to the variability of change-related effort across events. Observable dispersion measures approximate moments of an underlying difficulty distribution rather than directly instantiating the graph-theoretic variance term. 
A decrease in $U_t$ indicates that changes become more predictable, even if their average cost remains substantial.

The condition
\[
\Delta U_t < 0
\]
thus corresponds to a reduction in the dispersion of change effort over time. 
This is the empirical signature of stabilization in the burden--uncertainty framework.

\subsection{Empirical Meaning of the Structural Assumptions}

The assumptions A1--A4 introduced in Section~3 can also be interpreted empirically.

\begin{itemize}
\item \textbf{A1} corresponds to the persistence of structural load: changed entities remain embedded in dependency neighborhoods whose average exposure does not decrease.
\item \textbf{A2} corresponds either to structural regularization of the changed regions or to regularization in the effective distribution of changed targets, so that structural exposure becomes less heterogeneous over time.
\item \textbf{A3} corresponds to process stabilization: residual variability decreases due to learning, tooling, standardization, and repeated workflows.
\item \textbf{A4} corresponds to covariance control: structurally heavy changes become less likely to carry disproportionately large residual difficulty.
\end{itemize}

These assumptions are not arbitrary mathematical devices. 
Rather, they describe patterns that are plausible in mature software projects and consistent with empirical observations in large-scale development.

\subsection{Consistency with Observed Software Evolution}

Empirical studies of mining software repositories and change-based analysis suggest that change processes become increasingly structured over time~\cite{zimmermann2005,hassan2009}. 
At the same time, the systems under study do not necessarily exhibit a reduction in complexity or change burden.

This combination of persistent burden and decreasing variability is precisely the local regime characterized by Proposition~\ref{prop:local-sws}:
\[
\Delta B_t \geq 0,
\qquad
\Delta U_t < 0.
\]

Moreover, the proposition clarifies that decreasing uncertainty need not arise from a single source. 
It may result from reduced structural heterogeneity, reduced residual process noise, reduced covariance between structural exposure and residual difficulty, or any combination of these effects.

Accordingly, the theoretical model does not directly validate empirical proxies by itself. Rather, it provides a structural interpretation for why such observations can arise under identifiable conditions, while leaving the proxy-to-construct mapping to empirical measurement design.

\subsection{Scope of the Empirical Connection}

The empirical relation established here is intentionally modest. 
The goal is not to claim that every software system satisfies assumptions A1--A4, nor that all empirical cases must exhibit the same dynamic pattern.

Instead, the point is that the burden--uncertainty framework captures a meaningful class of empirical phenomena that one-dimensional views struggle to explain. 
Where empirical studies observe systems becoming more predictable without becoming simpler, the present model provides a principled interpretation.

A broader empirical evaluation across additional software ecosystems remains an important direction for future work.
\section{Discussion}

\subsection{Reframing Software Evolution}

The framework developed in this paper suggests that software evolution should not be understood solely as a process of increasing or decreasing complexity. 
Instead, it should be understood as a two-dimensional process in which structural burden and uncertainty evolve according to distinct dynamics.

This reframing is important because one-dimensional perspectives tend to treat complexity growth as inherently destabilizing. 
The present model shows that this is not necessarily the case: structural burden may persist or increase while uncertainty decreases.

\subsection{Theoretical Meaning of Stabilization}

Within the proposed framework, stabilization is defined not by a reduction in structural complexity, but by a reduction in the variability of change effort. 
This shifts the interpretation of software stability from simplification to predictability.

Under Proposition~\ref{prop:local-sws}, stabilization without simplification is not merely a heuristic intuition. 
It is a locally characterized regime represented by explicit assumptions about average structural load, structural regularization, process stabilization, and covariance control.

This result gives formal meaning to the idea that mature software systems may remain complex while becoming increasingly manageable.

\subsection{Relation to Existing Views of Complexity}

Traditional discussions of software evolution often emphasize complexity growth, technical debt, maintainability decline, or defect risk. 
These perspectives remain important, but they typically do not distinguish between the average cost of change and the variability of that cost.

The burden--uncertainty framework complements such views by introducing a second dimension. 
In this perspective, the problem is not only how much complexity a system accumulates, but also how predictably that complexity behaves under change.

This distinction helps explain why structural growth and practical stability are not always in contradiction.

\subsection{Relation to the Empirical Companion Study}

The theoretical model presented here is closely related to a separate empirical study of software evolution. 
That empirical study analyzes longitudinal data from multiple open-source software systems and reports cases in which uncertainty decreases over time while structural burden remains nontrivial.

The present paper provides the theoretical counterpart to that empirical finding. 
Its role is not to reanalyze the data, but to show that such a pattern is structurally coherent and representable under explicit assumptions.

Taken together, the two works provide a two-layered contribution:

\begin{itemize}
\item the empirical study identifies the pattern in real software systems, and
\item the present paper explains how such a pattern can arise as a consequence of burden--uncertainty decoupling.
\end{itemize}

Thus, the theoretical and empirical studies are complementary rather than redundant.

\subsection{Implications for Software Engineering}

The proposed framework suggests several implications for software engineering.

First, efforts to improve software stability should not be equated exclusively with efforts to reduce structural complexity. 
Systems may become more stable not because they become simpler, but because their change processes become more predictable.

Second, engineering practices such as testing, review discipline, workflow standardization, and accumulated developer knowledge may play a central role in stabilization by reducing uncertainty rather than by eliminating burden.

Third, the burden--uncertainty distinction provides a more precise language for discussing the evolution of large software systems. 
It allows stability and complexity to be analyzed as related but distinct phenomena.

Fourth, the introduction of a covariance term suggests that highly exposed structural regions may require special attention not only because they are costly, but because they may also amplify residual unpredictability. 
This points to a practical distinction between reducing average burden and reducing the exceptional difficulty associated with structurally sensitive changes.

\subsection{Limitations of the Framework}

The present model is intentionally minimal. 
It represents software structure through a dependency graph and change effort through a linear decomposition into structural and stochastic components.

This abstraction omits many aspects of real software systems, including semantic structure, organizational constraints, developer networks, and nonlinear change propagation. 
In particular, the use of first-order local dependency exposure as the structural measure is a simplifying choice rather than a complete representation of architectural burden.

Accordingly, the model should not be interpreted as a complete theory of software evolution. 
Its contribution is narrower but more precise: it isolates a structural representation of how stabilization without simplification can occur locally.

\subsection{Future Directions}

Several extensions follow naturally from this framework.

First, the graph model can be enriched to incorporate additional structural properties beyond local degree, such as transitive reachability, centrality, modularity, or multilayer dependency relations.

Second, the stochastic component can be generalized beyond variance-based descriptions to include heavier-tailed or non-Gaussian forms of uncertainty.

Third, the covariance structure between structural exposure and residual effort can be modeled more explicitly, rather than treated through monotonic control assumptions alone.

Fourth, the assumptions A1--A4 can be tested empirically across a broader range of software ecosystems in order to determine how widely the stabilization-without-simplification regime applies.

More broadly, the present work opens the possibility of developing a richer theory of software evolution in which structural burden and uncertainty are treated as separate but interacting dimensions.

\subsection{Conclusion}

This paper has proposed a graph-based, discrete-time probabilistic framework for software evolution in which structural burden and uncertainty are explicitly separated. 
Within this framework, we have characterized stabilization without simplification as a structurally admissible local regime.

The central implication is that software systems need not become structurally simpler in order to become more stable. 
What matters is whether the variability of change effort contracts even when structural load persists.

By making this distinction explicit, the burden--uncertainty framework provides a theoretical basis for understanding a recurring phenomenon in software evolution and offers a foundation for further theoretical and empirical work.

% =========================
\bibliographystyle{plain}
\bibliography{references}

\end{document}